\documentclass{IEEEtran4PSCC}

\usepackage{cite}

\ifCLASSINFOpdf
   \usepackage[pdftex]{graphicx}
\else
   \usepackage[dvips]{graphicx}
\fi
\usepackage[cmex10]{amsmath}
\usepackage{amsthm}
\ifCLASSOPTIONcompsoc
 \usepackage[caption=false,font=normalsize,labelfont=sf,textfont=sf]{subfig}
\else
 \usepackage[caption=false,font=footnotesize]{subfig}
\fi
\usepackage{caption}
\usepackage[font=small,labelfont=bf]{caption}

\usepackage{pgfplots}
\usepackage{tikz-network}
\DeclareUnicodeCharacter{2212}{−}
\usepgfplotslibrary{groupplots,dateplot}
\usetikzlibrary{patterns,shapes.arrows}
\usetikzlibrary{pgfplots.statistics}
\pgfplotsset{compat=newest}

\usepackage{enumitem}

\DeclareMathAlphabet{\mymathbb}{U}{BOONDOX-ds}{m}{n}

\newcommand{\bb}[1]{\mymathbb{#1}}

\newcommand{\mc}[1]{\mathcal{#1}}
\newcommand{\mt}[1]{\textnormal{#1}}

\newcommand{\lt}{\left}
\newcommand{\rt}{\right}
\newcommand{\beeq}{\begin{equation}}
\newcommand{\eneq}{\end{equation}}
\newcommand{\matb}{\begin{matrix}}
\newcommand{\mate}{\end{matrix}}

\newcounter{asses}
\newtheorem{assumption}[asses]{Assumption}
\newcounter{rems}
\newtheorem{remark}[rems]{Remark}
\newcounter{thms}
\newtheorem{lemma}[thms]{Lemma}

\makeatletter
\let\old@ps@headings\ps@headings
\let\old@ps@IEEEtitlepagestyle\ps@IEEEtitlepagestyle
\def\psccfooter#1{%
    \def\ps@headings{%
        \old@ps@headings%
        \def\@oddfoot{\strut\hfill#1\hfill\strut}%
        \def\@evenfoot{\strut\hfill#1\hfill\strut}%
    }%
    \def\ps@IEEEtitlepagestyle{%
        \old@ps@IEEEtitlepagestyle%
        \def\@oddfoot{\strut\hfill#1\hfill\strut}%
        \def\@evenfoot{\strut\hfill#1\hfill\strut}%
    }%
    \ps@headings%
}
\makeatother

\psccfooter{%
        \parbox{\textwidth}{\hrulefill \\ \small{23rd Power Systems Computation Conference} \hfill \begin{minipage}{0.2\textwidth}\centering \vspace*{4pt} \includegraphics[scale=0.06]{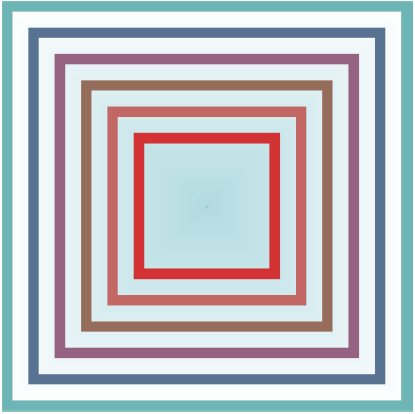}\\\small{PSCC 2024} \end{minipage} \hfill \small{Paris, France --- June 4 -- 7, 2024}}%
}

\begin{document}
%
\title{Power Grid Parameter Estimation Without Phase Measurements:\,Theory and Empirical Validation}

\author{
\IEEEauthorblockN{Jean-S\'ebastien Brouillon\\Giancarlo Ferrari Trecate}
\IEEEauthorblockA{Institute of Mechanical Engineering\\
Ecole Polytechnique F\'ed\'erale de Lausanne\\
Lausanne, Switzerland\\
\{jean-sebastien.brouillon, giancarlo.ferraritrecate\}@epfl.ch}
\and
\IEEEauthorblockN{Keith Moffat\\Florian D\"orfler}
\IEEEauthorblockA{Department of Information Technology and Electrical Engineering\\
ETH Z\"urich\\
Z\"urich, Switzerland\\
\{kmoffat, dorfler\}@ethz.ch}
}


\maketitle

\begin{abstract}
Reliable integration and operation of renewable distributed energy resources requires accurate distribution grid models. However, obtaining precise models is often prohibitively expensive, given their large scale and the ongoing nature of grid operations. To address this challenge, considerable efforts have been devoted to harnessing abundant consumption data for automatic model inference. The primary result of the paper is that, while the impedance of a line or a network can be estimated without synchronized phase angle measurements in a consistent way, the admittance cannot. Furthermore, a detailed statistical analysis is presented, quantifying the expected estimation errors of four prevalent admittance estimation methods. Such errors constitute fundamental model inference limitations that cannot be resolved with more data. These findings are empirically validated using synthetic data and real measurements from the town of Walenstadt, Switzerland, confirming the theory. The results contribute to our understanding of grid estimation limitations and uncertainties, offering guidance for both practitioners and researchers in the pursuit of more reliable and cost-effective solutions.
\end{abstract}


\thanksto{\noindent This research is supported by the Swiss National Science Foundation under the NCCR Automation (grant agreement 51NF40\_180545).}

\section{Introduction}

The deployment of sensors and machine learning techniques on power grids has opened a new set of power system monitoring applications. Among them, line parameter and topology estimation may play a crucial role for deploying smart energy resources \cite{IEEE_QER, OE_RD_Factsheet, NY_challenges_state_estim, Schenato, Iovine2019, aer2023flexible, liu2023project}. This is specially needed at the distribution level, as Distribution System Operators (DSOs) often lack accurate models of their Distribution Grids (DGs)\footnote{We focus on distribution network applications in this paper, however the analysis applies to transmission networks as well.}. Many recent studies have shown that physical models of a distribution grid can be estimated from synchronized voltage and current measurements \cite{tomlin, online_rls, cavraro_pmu, wehenkel2020parameter, moffat_2019}. However, micro synchrophasor measurement units ($\mu$PMUs) remain expensive and therefore mostly absent in DGs.

Sensor costs motivate grid analysis that does not rely on accurate phase angle measurements, as voltage and current magnitude sensors are considerably cheaper than $\mu$PMU sensors. Smart meter sensors \cite{smart_meters_characteristics} are becoming ubiquitous for measuring power consumption. The load—coverage in European DGs is already above 70\% \cite{smart_meter_rollout}.
Thus, the question whether it is possible to accurately estimate grid impedance or admittance without synchronized phase angle measurements is of practical importance. This paper investigates the feasibility of impedance and admittance estimation from voltage and current magnitude and power angle measurements only.

Given that phase angles are typically small in DGs, one might consider applying the methods used for synchronized measurements with a zero angle. However, recent research demonstrates that such an approach results in an inconsistent estimate \cite{lisa2023}, even when employing Error-in-Variables (EIV) methods, e.g., the Total Least Squares (TLS). 
Consistency is critical when estimating DG parameters because the stability of the voltage operating point leads to a low Signal-to-Noise Ratio (SNR), which can only be mitigated by collecting a large volume of data.


Prior research has improved the accuracy of power system parameter estimation from smart meter measurements. The authors of \cite{zhang_non_pmu_newton} attempted to simultaneously estimating phase angles and parameters, which reduced the estimation error in some cases but did not provide consistency guarantees. A similar approach is employed in \cite{mixed_pmu_sm_newton}, where $\mu$PMUs are added at some nodes to mitigate inconsistencies, albeit not eliminating them entirely. A more recent development \cite{zhang_non_pmu_id} focuses on canceling the phase variable out of the current flow equations. This reduction yields consistent estimates of X/R ratios and certain transformations of conductance and susceptance, but not of these quantities themselves. Alternatively, some studies such as \cite{line_length_vanin} have chosen to estimate impedance rather than admittance, but inverting the impedance estimates also results in accuracy issues \cite{moffat_2019}.


In this paper, we aim to enhance our understanding of existing estimation methods by quantifying their inherent biases, i.e., their expected estimation errors. Our contribution can be summarized in three main aspects. We first present a detailed statistical model for common DG sensors and utilize it to express the biases associated with each estimation method in a single-line identification setup. Notably, we reveal that while each method effectively addresses certain biases from previous approaches, no admittance estimation method achieves a perfect bias cancellation, for which synchronized phasor measurements are required. Second, we show how the same reasoning can be applied to admittance matrix estimation problems when the topology is unknown. Third, to validate our findings and underscore the practical challenges of the problem, we conduct an experiment using real DG data, thereby demonstrating the applicability of the methods in the real world. 



\subsection{Preliminaries and Notation}
Complex numbers are defined as $z = a + j b$, where $a,b \in \bb R^2$ and $j^2 = -1$. The complex conjugate of $z \in \bb C$ is denoted by $z^\star$. The pseudo-inverse $A^\dagger$ of a matrix $A \in \bb C^{m \times n}$ is obtained by inverting its non-zero singular values, i.e., if $A = U \mt{diag}([s, 0,\dots, 0]) (V^\star)^\top$ then $A^\dagger = U \mt{diag}([1/s, 0,\dots, 0]) (V^\star)^\top$. If $A$ is square and invertible, then $A^\dagger$ is equal to $A^{-1}$ the inverse of $A$. The noisy measurement of a variable $x$ is denoted by $\tilde x$. In a regression model $\tilde z = A \tilde x + \epsilon$, we call $\tilde x$ and $\tilde z$ the Right-Hand Side (RHS) and Left-Hand Side (LHS) variables, respectively. All variables have an implicit time-dependence, except the parameters of the network and probability distributions. This dependence is implicit. The stacked vector of $N+1$ measurements of a variable $x$ is denoted by $[\tilde x]_{t=t_0}^{t_N}$

\section{System model}
\subsection{Distribution Network}\label{subsec_model}
We model the a power grid as a graph $\mc G(\mc V, \mc E)$ with $n$ nodes $\mc V = \{1, \dots, n\}$ and edges $\mc E \subseteq \mc V \times \mc V$. Each node $h = 1,\dots, n$ is injecting a current $i_h$ at a voltage $v_h$. Consumer nodes are modeled with a negative injection. Nodal AC voltages can be expressed as phasors $v_h = |v_h|e^{j\theta_h}$, where $\theta_h$ is the phase difference with node $1$. The nodal power angle $\phi_{h}$ is the phase difference between the current injection $i_{h}$ and the voltage $v_h$ (i.e., the arc cosine of the power factor). Additionally, we define all the quantities for each line $h \rightarrow k$ using a double subscript notation "$hk$". Thus, the voltage drop, current flow, and power angle of a line from $h$ to $k$ are $v_{hk}$, $i_{hk}$, and $\phi_{hk}$, respectively. For convenience, we define $\phi_{hk}$ as the angle between $i_{hk}$ and $v_h$.

\begin{assumption}\label{ass_small_theta}
    The voltage phase angle differences are small, i.e., $|\theta_h| \ll 1$.
\end{assumption}
Comparatively small amounts of power is transmitted on DGs, which are also quite resistive. This means that Assumption \ref{ass_small_theta} is quite mild. Similar to the voltages, nodal currents are also phasors $i_h = |i_h|e^{j\theta_h - j\phi_h}$, where the nodal power angle $\phi_h$ is not small in general.

We use a lumped-$\pi$ circuit to model each electrical connection in $\mc E$ \cite{kundur2007power}, where the lines are modeled as an inductor and a resistor in series, and the shunts as capacitors. This gives the line and shunt admittances $y_{hk} = g_{hk} + jb_{hk}$ and $y_{hh,s}$, respectively. If two nodes $h$ and $k$ are not connected, we use $y_{hk} = 0$. We assume that there are no phase-shifting transformers, and that the existing transformers have a constant ratio over the duration of the experiment. This means that the voltages and currents can be re-scaled to model the transformers as simple line admittances.

All parameters are collected in the admittance matrix $Y = G+jB$ with $G$ and $B$ in $\bb R^{n\times n}$. $Y$ is defined by its non-diagonal elements $Y_{hk} = y_{hk}, \forall h,k \in \mc V^2$ and its diagonal ones $Y_{kk} = -\sum_{k=1}^n y_{hk} + y_{hh,s}, \forall h \in \mc V$. For convenience, we also define the X/R ratio $\rho_{hk} = -\frac{b_{hk}}{g_{hk}}$ of each line\footnote{Inverting a complex number does not invert the ratio between its real and imaginary parts.}. 

\subsection{Measuring Devices}\label{subsec_sensors}
Smart meters provide the measurements $|\tilde v_h|$ and $|\tilde i_h|$ of the amplitudes $|v_h|$ and $|i_h|$ of the voltage and current and the measurement $\tilde \phi_h$ of the power angle $\phi_h$. Moreover, such sensors can provide the line flow measurements $\tilde i_{hk}$ and $\tilde \phi_{hk}$ if they are placed on a specific line rather than a node.
\begin{assumption}\label{ass_gaussian}
    The noise on the current magnitude, voltage magnitude, and power angle measurements is independent, Gaussian, and centered on zero.
\end{assumption}
Although the regulations of commercial smart meters require a given maximum admissible error, which would imply that the noise follow a truncated Gaussian \cite{smart_meter_classes}, the resulting interval is usually large enough to approximate the noise distribution as Gaussian.
\begin{assumption}\label{ass_small_deltaphi}
    The error on power angle measurement is small.
\end{assumption}
With Assumption \ref{ass_small_deltaphi}, the noise on the current phasor is almost perfectly Gaussian in Cartesian coordinates. However, the real and imaginary parts of the noise on current measurements are not independent if the power angle is non-zero.
\begin{remark}
    The active and reactive powers are not a linear combination of the measured quantities and are subject to the dependence between the real and imaginary parts of the current noise. Their uncertainty is therefore neither Gaussian-distributed nor independent, which is a common assumption in the literature.
\end{remark}

For clarity, we define the noise on the difference of voltage magnitudes between two nodes and its distribution as
\begin{subequations}\label{eq_stats_noise}
\begin{align}\label{eq_stats_noise_v}
    \delta_{hk}^v := (|\tilde v_h| - |\tilde v_k|) - (|v_h| - |v_k|) \sim \mc{N}(0, \sigma_{hk}^v).
\end{align}
We chose the lower-case symbol $\sigma_{hk}^v$ for the variance to emphasize the fact that it is scalar. We also define the noise on the line current flow $|\tilde i_{hk}| e^{-j \tilde \phi_{hk}}$ measured at node $h$ as
\begin{align}\label{eq_stats_noise_i}
    \!\lt[\matb \delta_{hk}^\Re \\ \delta_{hk}^\Im \mate\rt] \!\!:=\!\! 
    \lt[\matb |\tilde i_{hk}|\cos(\tilde \phi_{hk}) \!-\! |i_{hk}|\cos(\phi_{hk}) 
    \\
    |\tilde i_{hk}|\sin(\tilde \phi_{hk}) \!-\! |i_{hk}|\sin(\phi_{hk})
    \mate\rt]
    \!\sim \mc{N}(0, \Sigma_{hk}^i),\!
\end{align}
\end{subequations}
where $\Sigma_{hk}^i$ is a $2$-by-$2$ symmetric matrix containing the variances $\sigma_{hk}^\Re = \mt{var}[\delta_{hk}^\Re]$ and $\sigma_{hk}^\Im = \mt{var}[\delta_{hk}^\Im]$ of both parts of the noise on $i_{hk}$ in the diagonal elements, and their covariance $\mt{cov}[\delta_{hk}^\Re, \delta_{hk}^\Im]$ in the off-diagonal elements. Finally, we consider the the final measurements are synchronized block-averages of the instantaneous readings, and that the average over the estimation window is subtracted from each block to center the data.

\section{Line parameter estimation}\label{section_one_line}


We first study the estimation problem for a single line before addressing network estimation in Section \ref{section_network}.

\subsection{Impedance Estimation}
The parameters $z_{hk} = y_{hk}^{-1}$ of a line $h\rightarrow k$ relate the current flow $i_{hk}$ to the difference of voltages $v_h - v_k$ at each end of the line according to the following relationship
\begin{align}\label{eq_line_flow_z_og_complex}
     |v_h|e^{j\theta_h} - |v_k|e^{j\theta_k} &= z_{hk} |i_{hk}|e^{j(\theta_h - \phi_{hk} )},
\end{align}
which can be written in real numbers and linearized using $e^{-j\theta_{hk}} \approx 1 - j\theta_{hk}$ from Assumption \ref{ass_small_theta} as
\begin{subequations}\label{eq_line_flow_z_lin_real}
\begin{align}\label{eq_line_flow_z_lin_real_a}
    |v_h| - |v_k| &= r_{hk}|i_{hk}|\cos(\phi_{hk}) + x_{hk} |i_{hk}|\sin(\phi_{hk}), \\
    |v_k|\theta_{hk} &= -r_{hk}|i_{hk}|\sin(\phi_{hk}) + x_{hk}|i_{hk}|\cos(\phi_{hk}).
    \label{eq_line_flow_z_lin_real_b}
\end{align}
\end{subequations}
As shown in \cite{lisa2023}, \eqref{eq_line_flow_z_lin_real} is a very close approximation of \eqref{eq_line_flow_z_og_complex} when the currents and voltages are exact. 

The phase $\theta_{hk}$ is not measured so one must use \eqref{eq_line_flow_z_lin_real_a} to fit the parameters $r_{hk}$ and $x_{hk}$. With noisy measurements, the regression is performed on the model
\begin{align}\label{eq_line_flow_z_lin_real_reduced_simple}
    \!\!|\tilde v_h| \!-\! |\tilde v_k| &= r_{hk} |\tilde i_{hk}|\cos(\tilde \phi_{hk})
    + x_{hk} |\tilde i_{hk}|\sin(\tilde \phi_{hk}) +\! \epsilon_{hk},\!
\end{align}
where $\epsilon_{hk} = \delta_{hk}^v - r_{hk} \delta_{hk}^\Re - x_{hk} \delta_{hk}^\Im$ embeds the uncertainty. 

\begin{remark}\label{rem_regu}
    Under Assumption \ref{ass_small_theta}, \eqref{eq_line_flow_z_lin_real_b} $\approx 0$ so one can also use \eqref{eq_line_flow_z_lin_real_b} as a regularizer with strength $\lambda$,~by adding $-\lambda r_{hk}|\tilde i_{hk}|\sin(\tilde \phi_{hk}) \!+\! \lambda x_{hk}|\tilde i_{hk}|\cos(\tilde \phi_{hk}) \approx 0$ to the regression. This can be done by augmenting the voltage and current data matrices with $[0]_{t=t_0}^{t_N}$ and $[-\lambda |\tilde i_{hk}|\sin(\tilde \phi_{hk}), \lambda |\tilde i_{hk}|\cos(\tilde \phi_{hk})]_{t=t_0}^{t_N}$, respectively.
\end{remark}

\subsection{Admittance Estimation}
Similar to the impedance, the admittance relates the current flow $i_{hk}$ to the difference of voltages $v_h - v_k$ at each end of the line, but in the inverse way. This gives following relationship
\begin{align}\label{eq_line_flow_og_complex}
    |i_{hk}|e^{-j\phi_{hk}}
    &= y_{hk}\lt(|v_h| - |v_k|e^{-j\theta_{hk}}\rt).
\end{align}
As for \eqref{eq_line_flow_z_lin_real}, using the truncated expansion $e^{-j\theta_{hk}} \approx 1 - j\theta_{hk}$ yields
\begin{subequations}\label{eq_line_flow_lin_real}
\begin{align}
    |i_{hk}|\cos(\phi_{hk}) &= g_{hk}(|v_h| - |v_k|) - b_{hk}|v_k|\theta_{hk}, \\
    |i_{hk}|\sin(\phi_{hk}) &= - g_{hk}|v_k|\theta_{hk} - b_{hk}(|v_h| - |v_k|).
\end{align}
\end{subequations}
One can observe that the unobserved phase angle $\theta_{hk}$ is not simple from \eqref{eq_line_flow_lin_real} as it is done in \eqref{eq_line_flow_z_lin_real}. Assumption \ref{ass_small_theta} hints that one could handle the missing phase measurements by replacing them with zero. However, the phase often has a key role in power transmission in practice so assuming it to be exactly zero is often too constraining. Nevertheless, for very resistive grids, inferring $g_{hk}$ and $b_{hk}$ in the following regression model sometimes yields good results.
\begin{subequations}\label{eq_simple_ls_regression_noisy}
\begin{align}\label{eq_simple_ls_regression_noisy_a}
    \!\!\!|\tilde i_{hk}|\cos(\tilde \phi_{hk}) \!&= g_{hk}(|\tilde v_h| \!-\! |\tilde v_k|) + \mu_{hk} \!-\! b_{hk}|v_k|\theta_{hk},\! \\
    \!\!\!|\tilde i_{hk}|\sin(\tilde \phi_{hk}) \!&= -b_{hk}(|\tilde v_h| \!-\! |\tilde v_k|) \underbrace{+\, \nu_{hk} \!-\! g_{hk}|v_k|\theta_{hk}}_{\mt{considered as noise}},\!
    \label{eq_simple_ls_regression_noisy_b}
\end{align}
\end{subequations}
where $\mu_{hk} = \delta_{hk}^\Re - g_{hk} \delta_{hk}^v$ and $\nu_{hk} =  \delta_{hk}^\Im + b_{hk}\delta_{hk}^v$.

To conclude this section, we observe that the admittance regression model \eqref{eq_simple_ls_regression_noisy} actively enforces a small phase angle $\theta_{hk}$, in contrast with the impedance regression \eqref{eq_line_flow_z_lin_real_reduced_simple}.



\section{Impedance Estimation Bias and Variance} 

This section characterizes the bias and variance of single line impedance estimation.
As described in Appendix \ref{app_TLS}, fitting \eqref{eq_line_flow_z_lin_real_reduced_simple} using the Ordinary Least Squares (OLS) results in biased estimates of $x_{hk}$ and $r_{hk}$. This bias can be eliminated by using Total Least Squares (TLS). However, if there are correlations between the LHS and the RHS of the regression model, the TLS estimates can also be biased. Nevertheless, for the noise model \eqref{eq_stats_noise}, the TLS are a close approximation of the Maximum Likelihood Estimator (MLE), which can only be computed exactly if the phase is known.

\begin{lemma}\label{lem_bias}
    If a regression model $\tilde z = A (\tilde x - \epsilon_x) + \epsilon_z$, where $\epsilon_x, \epsilon_x \sim \mc N(0, \Sigma)$\footnote{The RHS and LHS variables must be normalized to have the same noise variance.} is fitted to $D$ datasets $\tilde x_d, \tilde z_d$, the bias of the TLS is given by
\begin{align}\label{eq_tls_bias}
    \!E[\hat A] \!-\! A = &\!\lt(\sum_{d=1}^D \mt{var}[\tilde x_d] - \mt{var}[\epsilon_{xd}] \!\rt)^{\!\!\!-1}
    \\ \nonumber
    &
    \cdot\!\! \lt(\sum_{d=1}^D \mt{cov}[\epsilon_{xd}, \epsilon_{zd}] + \mt{cov}[\epsilon_{xd}, \tilde z_d] + \mt{cov}[\tilde x_d, \epsilon_{zd}] \!\rt) \!\!.
\end{align}
\end{lemma}
\begin{proof}
    The proof is derived in Appendix \ref{app_tls_bias}.
\end{proof}

Applying Lemma \ref{lem_bias} with $D=1$ to the model \eqref{eq_line_flow_z_lin_real_reduced_simple} with the noise statistics \eqref{eq_stats_noise} gives
the bias
\begin{align*}
    &E[\hat r_{hk}] - r_{hk} = (\mt{var}[|\tilde i_{hk}|\cos(\tilde \phi_{hk})] - \mt{var}[\delta_{hk}^\Re])^{-1}
    \\ \nonumber
    &
     \!\cdot \! (\mt{cov}[\delta_{hk}^\Re, |\tilde v_h| \!\!-\!\! |\tilde v_k|] \!+\! \mt{cov}[|\tilde i_{hk}| \! \cos(\tilde \phi_{hk}), \delta_{hk}^v] \!-\! \mt{cov}[\delta_{hk}^\Re, \delta_{hk}^v]),
\end{align*}
and similarly for $x_{hk}$. All three covariances are zero in \eqref{eq_stats_noise} so we conclude that the TLS estimate of the impedance is unbiased. Moreover, the variance of the impedance estimate can be derived from \cite{crassidis_tls_cov} as\footnote{$\mt{var}[\hat r_{hk}, \hat x_{hk}]$ depends on the exact parameters $r_{hk}$ and $x_{hk}$ and can, in practice, only be estimated if the SNR is sufficiently high \cite{crassidis_tls_cov}.}
\begin{align}\label{eq_rx_variance}
    \!\!\mt{var}\lt[\matb \hat r_{hk}
    \\
    \hat x_{hk} \mate\rt] &=
    \frac{\sigma_{hk}^v \!+\! \|[r_{hk}, x_{hk}]\|^2_{\Sigma_{hk}^i}}{N}
    \mt{var} \!
    \lt[\matb
    |\tilde i_{hk}|\cos(\tilde \phi_{hk})
    \\
    |\tilde i_{hk}|\sin(\tilde \phi_{hk})
    \mate\rt]^{-1}\!\!\!\!,
\end{align}
where $\|[r_{hk}, x_{hk}]\|^2_{\Sigma_{hk}^i} \!\! = [r_{hk}, x_{hk}]\Sigma_{hk}^i [r_{hk}, x_{hk}]^\top$. Equation \eqref{eq_rx_variance} shows the consistency\footnote{Consistent means that it converges with probability $1$ to the exact parameters as $N\rightarrow\infty$. Efficient means that its variance is equal to the Camer-Rao lower bound\cite{statistical_inference_book}.} of the TLS estimate because $\mt{var}[\hat r_{hk}, \hat x_{hk}] \rightarrow 0$ as $N \rightarrow \infty$.


\section{Admittance estimation bias}\label{section_biases}

This section characterizes the biases of several admittance estimation methods from the literature for a single line. We focus on the limit case where the sample size is very large, because the amount of data is not a big limitation for the estimation problem. This means that although the variance can be derived from \cite{crassidis_tls_cov} (similar to \eqref{eq_rx_variance}), we do not include it in this study because it decays to zero as $N \rightarrow \infty$ and is therefore less limiting than the biases.

The analysis demonstrates that four known methods for estimating the admittance from smart meter measurements produce biased estimates. Nevertheless, we show in Section \ref{subsec_xr_regression} that, surprisingly, the biases decay with $N \rightarrow \infty$ only if one inverts the impedance estimate. Thus, if one is interested in the admittance of a line from a large data set, it may be better to estimate the impedance and then invert the estimate.

\subsection{Omitted-Phase Bias}\label{subsec_omission}

We start by quantifying the bias that appears when fitting the model \eqref{eq_simple_ls_regression_noisy} using the TLS. The magnitude difference $|\tilde v_h| - |\tilde v_k|$ is often correlated with the phase shift $\theta_{hk}$, which means that a bias may appears if it is incorporated as noise, as explained in Section \ref{subsec_endo}. Using Lemma \ref{lem_bias} with $D=1$ and the model \eqref{eq_simple_ls_regression_noisy_a}, we can express the resulting omitted-phase bias\footnote{More commonly known as omitted-variable bias.} with
\begin{align*}
E[\hat g_{hk}] &\!-\! g_{hk} = \\ \nonumber
&\; b_{hk}(\mt{var}[|\tilde v_h| \!-\! |\tilde v_k|] \!-\! \sigma_{hk}^v)^{-1} \mt{cov}[|\tilde v_h| \!-\! |\tilde v_k|,\! |v_k|\theta_{hk}],
\end{align*}
and similarly for $\hat b_{hk}$.

In order to avoid the omitted-phase bias, the authors of \cite{zhang_non_pmu_newton, mixed_pmu_sm_newton} estimate both the phase $\theta_{hk}$ and the parameters $g_{hk}$ and $b_{hk}$ iteratively in a joint problem. This means using the estimates $\hat g_{hk}$ and $\hat b_{hk}$ to find the estimate $\hat e_{hk}^\theta$ of $|v_k|\theta_{hk}$, which can in be plugged in \eqref{eq_simple_ls_regression_noisy} to update $\hat g_{hk}$ and $\hat b_{hk}$. Any method to compute $\hat e_{hk}^\theta$ will aim at satisfying \eqref{eq_simple_ls_regression_noisy}. Rearranging the terms, vectorizing and plugging in the estimates, this means
\begin{align}\label{eq_simple_ls_regression_noisy_phase}
    \hat e_{hk}^\theta \!\! \lt[\matb \hat b_{hk} \\ \hat g_{hk} \mate\rt] \!\!
    &= \!(|\tilde v_h| \!-\! |\tilde v_k| \!-\! \delta_{hk}^v)
    \!\!
    \lt[\matb \hat g_{hk} \\ -\hat b_{hk} \mate\rt] 
    \!\!-\!
    |\tilde i_{hk}| \!\! \lt[\matb \cos(\tilde \phi_{hk})
    \\
    \sin(\tilde \phi_{hk}) \mate\rt]
    \!\!+\!\!
    \lt[\matb \delta_{hk}^\Re \\ \delta_{hk}^\Im \mate\rt] \!\!.
\end{align}
The estimate $\hat e_{hk}^\theta$ can then be used in \eqref{eq_simple_ls_regression_noisy} as
\begin{subequations}
    \label{eq_simple_ls_regression_noisy_with_phase_est}
\begin{align}
    \label{eq_simple_ls_regression_noisy_with_phase_est_a}
    \!\!\!|\tilde i_{hk}|\cos(\tilde \phi_{hk}) \!&= [g_{hk}, b_{hk}] [(|\tilde v_h| \!-\! |\tilde v_k|), \!-\hat e_{hk}^\theta]^\top \!+ \mu_{hk} ,\! \\
    \!\!\!|\tilde i_{hk}|\sin(\tilde \phi_{hk}) \!&= -[g_{hk}, b_{hk}] [\hat e_{hk}^\theta, (|\tilde v_h| \!-\! |\tilde v_k|)]^\top \!+ \nu_{hk},\!
    \label{eq_simple_ls_regression_noisy_with_phase_est_b}
\end{align}
\end{subequations}

While estimating $\theta_{hk}$ avoids the omitted-phase bias, it leads to a real and imaginary simultaneity bias\footnote{The simultaneity bias problem is well known in statistics. It originates from the lack of causality between variables that occur simultaneously \cite{simultaneity_book}, i.e., variations in voltage and current flow.}, which is presented in the next section.

\subsection{Real and Imaginary Simultaneity}\label{subsec_simultaneity}

A simultaneity bias appears when LHS variables in one equation appear on the RHS of another equation, hence creating correlation between the noise on the LHS and the variables in the RHS. To show that this correlation exists, we consider that the estimates are good, i.e., $[\hat e_{hk}^\theta, \hat g_{hk}, \hat b_{hk}] \approx [|v_k|\theta_{hk}, g_{hk}, b_{hk}]$. Because $[\hat g_{hk},\hat b_{hk}] [g_{hk},b_{hk}]^\dagger \approx 1$ by construction, one obtains
\begin{align*}
    %
    \mt{cov}\!\big[[ \delta_{hk}^\Re, \delta_{hk}^\Im], \hat e^{\theta}_{hk}\big]\!
    &\approx
    \mt{cov}\!\big[[ \delta_{hk}^\Re, \delta_{hk}^\Im], \hat e^{\theta}_{hk}[\hat g_{hk}, \hat b_{hk}]\big]  [g_{hk},b_{hk}]^\dagger.
\end{align*}
The expression \eqref{eq_simple_ls_regression_noisy_phase} shows that $\hat e_{hk}^\theta$ is correlated to the noises $\delta_{hk}^\Re$ and $\delta_{hk}^\Im$. Hence,
\begin{align*}
    %
    \mt{cov}\!\big[[ \delta_{hk}^\Re, \delta_{hk}^\Im], \hat e^{\theta}_{hk}[\hat g_{hk},\hat b_{hk}]\big]
     =
     \mt{var}\!\!\lt[\matb \delta_{hk}^\Re \\ \delta_{hk}^\Im \mate\rt]\!,
\end{align*}
which means that
\begin{align*}
    \mt{cov}\!\big[[ \delta_{hk}^\Re, \delta_{hk}^\Im], \hat e^{\theta}_{hk}\big]\!
    &\approx \Sigma_{hk}^i [g_{hk},b_{hk}]^\dagger\!.
\end{align*}
Hence, using Lemma \ref{lem_bias} for \eqref{eq_simple_ls_regression_noisy_with_phase_est} with $D=2$ and the datasets $[|\tilde v_h| \!-\! |\tilde v_k|, -\hat e^{\theta}_{hk}], |\tilde i_{hk}|\cos(\tilde \phi_{hk})$ and $ [-\hat e^{\theta}_{hk}, -|\tilde v_h| \!+\! |\tilde v_k|], |\tilde i_{hk}|\sin(\tilde \phi_{hk})$ yields the real and imaginary simultaneity bias
\begin{align*}
    &E[\hat b_{hk}, \hat g_{hk}]^\top \!-\! [b_{hk}, g_{hk}]^\top \approx 
    \\ \nonumber &-\!\!
    \Bigg( \! \underbrace{\!\!
    \mt{var} \!\!\lt[\matb
    \!|\tilde v_h| \!-\! |\tilde v_k|\!
    \\
    -\hat e^{\theta}_{hk}
    \mate\rt]
    }_{\mt{from \eqref{eq_simple_ls_regression_noisy_with_phase_est_a}}}\!
    \!+\!
    \underbrace{\mt{var} \!\!\lt[\matb
    \hat e^{\theta}_{hk}
    \\
    \!|\tilde v_h| \!-\! |\tilde v_k|\!
    \mate\rt] 
    }_{\mt{from \eqref{eq_simple_ls_regression_noisy_with_phase_est_b}}}\!
    \!-\!\!
    \lt[\matb
    \!s_{hk}^e & \!\!\!\! 0
    \\
    0 & \!\!\!\! s_{hk}^e
    \mate\rt] \!
    \!\!\Bigg)^{\!\!\!\!-1}\!
    \Sigma_{hk}^i [g_{hk},b_{hk}]^\dagger\!,
\end{align*}
where $s_{hk}^e = \mt{var}[\hat e^{\theta}_{hk} \!-\! |\tilde v_k|\theta_{hk} ] + \sigma^v_{hk}$. 

The solution given in the literature \cite{simultaneity_book} to the simultaneity bias is to study the reduced regression model, i.e. plugging one equation into the other, substituting the exact unknown variable rather than its estimate. This has been done in \cite{zhang_non_pmu_id}, where the authors reduce the equations to remove the unobserved phase $\theta_{hk}$. We proceed similarly here, i.e., reducing $|\tilde v_k|\theta_{hk}$ by plugging \eqref{eq_simple_ls_regression_noisy_b} into \eqref{eq_simple_ls_regression_noisy_a}, which gives the reduced regression model
\begin{align}\label{eq_line_flow_lin_real_reduced}
    |\tilde i_{hk}|\cos(\tilde \phi_{hk}) &= - \rho_{hk}|\tilde i_{hk}|\sin(\tilde \phi_{hk}) \\ \nonumber
    &\!\!\!\!\!\!\! + g_{hk}(1 \!+\! \rho_{hk}^{2})(|\tilde v_h| \!-\! |\tilde v_k| - \delta_{hk}^v) + \rho_{hk} \delta_{hk}^\Im + \delta_{hk}^\Re.
\end{align} 
This regression gives the estimates of $\rho_{hk}$ and $g_{hk}(1 \!+\! \rho_{hk}^{2})$. The latter can then be divided by $(1+\hat \rho_{hk}^2)$ to obtain an estimate of $g_{hk}$. A similar process can be followed for $b_{hk}$ by plugging \eqref{eq_simple_ls_regression_noisy_a} into \eqref{eq_simple_ls_regression_noisy_b}.

\subsection{Endogeneity of $\delta_{hk}^\Re$ and $\delta_{hk}^{\Im}$}\label{subsec_endo}

We investigate the bias of TLS estimates when using the regression model \eqref{eq_line_flow_lin_real_reduced}. As explained in Section \ref{subsec_sensors}, there is a correlation between $\delta_{hk}^\Re$ and $\delta_{hk}^\Im$ when the power angle is not zero. Hence, from Lemma \ref{lem_bias}, the estimate of $\hat \rho_{hk}$ has the following endogeneity bias
\begin{align*}
    &E[\hat \rho_{hk}] \!-\! \rho_{hk} = [1, 0] \,\cdot 
    \\ \nonumber &\quad\quad
    -\!\!\lt( \!\!\mt{var}\!\lt[\matb
    |\tilde i_{hk}|\sin(\tilde \phi_{hk})
    \\
    |\tilde v_h| \!-\! |\tilde v_k|
    \mate\rt]\! \!-\! \!\lt[\matb
    \sigma_{hk}^\Im &\!\!\!\! 0
    \\
    0 &\!\!\!\! \sigma_{hk}^v
    \mate\rt]\!\rt)^{\!\!\!-1}\!\!
    \lt[\matb \mt{cov}[\delta_{hk}^\Re, \delta_{hk}^\Im] \\ 0 \mate\rt] \!\!,
\end{align*}
when estimated from \eqref{eq_line_flow_lin_real_reduced}. 

The removal of the endogeneity bias can be done by dividing \eqref{eq_line_flow_lin_real_reduced} by $g_{hk}(1 \!+\! \rho_{hk}^{2})$ and putting $|\tilde v_h| - |\tilde v_k|$ on the LHS. Thus, the dependence of $\delta_{hk}^\Re$ and $\delta_{hk}^\Im$ does not cross sides. Surprisingly, adjusting \eqref{eq_line_flow_lin_real_reduced} in this way produces \eqref{eq_line_flow_z_lin_real_reduced_simple}, an impedance estimate, because
\begin{subequations}\label{eq_rx_def}
    \begin{align}
        r_{hk} &= \frac{g_{hk}^{-1}}{(1 \!+\! \rho_{hk}^{2})} = \lt(g_{hk} \!+\! b_{hk}^2 g_{hk}^{-1}\rt)^{-1}\!\!,
        \\
        x_{hk} &= \rho_{hk} r_{hk} = -\lt(b_{hk} \!+\! g_{hk}^2 b_{hk}^{-1}\rt)^{-1}\!\!.
    \end{align}
\end{subequations}

\subsection{Impedance Inversion Bias}\label{subsec_xr_regression}

If the goal is to estimate the admittance values, rather than the impedance values, the problem is still not fully solved as one must invert \eqref{eq_rx_def} to obtain the estimates of $g_{hk}$ and $b_{hk}$\footnote{In Section \ref{subsec_simultaneity}, inverting $(1+\hat \rho_{hk}^2)$ also leads to a bias because $\rho_{hk}$ is not known exactly.}. This yields
\begin{subequations}\label{eq_gb_def}
    \begin{align}
        g_{hk} &= \lt(r_{hk} \!+\! x_{hk}^2 r_{hk}^{-1}\rt)^{-1},
        \\
        b_{hk} &= -\lt(x_{hk} \!+\! r_{hk}^2 x_{hk}^{-1}\rt)^{-1}.
    \end{align}
\end{subequations}
We can quantify the quality of the admittance estimate using bounds on $E[|\hat g_{hk}|]$ and $E[|\hat b_{hk}|]$, as the signs are known\footnote{From the line model in Section \ref{subsec_model}, $r_{hk}$ and $x_{hk}$ are positive. One can enforce the signs of the conductance and susceptance by setting $\hat g_{hk} = 0$ if the result was negative and $\hat b_{hk} = 0$ if it was positive.}. To do so, we rewrite \eqref{eq_gb_def} as
\begin{align*}
|g_{hk}| &= r_{hk}(x_{hk}^2 + r_{hk}^2)^{-1} + x_{hk} \cdot 0, \\
|b_{hk}| &= r_{hk} \cdot 0 + x_{hk}(x_{hk}^2 + r_{hk}^2)^{-1}.
\end{align*}
Plugging the functions $f_1(r,x) = (x^2 + r^2)^{-1}$ and $f_2(r,x) = 0$, which are convex on the positive quadrant, into a sharp\footnote{the equality holds for at least one realization of $\hat r_{hk}$ and $\hat x_{hk}$} Jensen-like inequality \cite[Theorem 2.1]{multivariate_jensen} gives the bounds\footnote{Upper bounds were recently discovered for univariate functions \cite{jensen_sharp}, however there are no upper bounds yet for bivariate problems such as \eqref{eq_gb_def}}
\begin{align*}
E[|\hat g_{hk}|] \!&\geq\! E[\hat r_{hk}] f_1 \!\!\lt(\! 
\frac{\mt{var}[\hat r_{hk}]}{E[\hat r_{hk}]} \!+\! E[\hat r_{hk}]
,\!
\frac{\mt{cov}[\hat x_{hk} ,\! \hat r_{hk}]}{E[\hat r_{hk}]} \!+\! E[\hat x_{hk}] 
\!\!\rt)\!\!,
\\
E[|\hat b_{hk}|] \!&\geq\! E[\hat x_{hk}] f_1 \!\!\lt(\! 
\frac{\mt{cov}[\hat x_{hk} ,\! \hat r_{hk}]}{E[\hat x_{hk}]} \!+\! E[\hat r_{hk}]
,\!
\frac{\mt{var}[\hat x_{hk}]}{E[\hat x_{hk}]} \!+\! E[\hat x_{hk}]
\!\!\rt)\!\!,
\end{align*}
where the variance of the estimates $\mt{var}[\hat r_{hk}, \hat x_{hk}]$ is given by \eqref{eq_rx_variance}. Using the definition of $f_1(r,x) = (r^2 + x^2)^{-1}$ and noting that $E[\hat r_{hk}] = r_{hk}$ and $E[\hat x_{hk}] = x_{hk}$, one gets
\begin{subequations}\label{eq_gb_expect_def}
\begin{align}
\!\!\!\!E[|\hat g_{hk}|] \!&\geq\! \frac{r_{hk}^3}{(\mt{var}[\hat r_{hk}] \!+\! r_{hk}^2)^2 \!+\! (\mt{cov}[\hat x_{hk}, \hat r_{hk}] \!+\! r_{hk}x_{hk})^2},\!\!\!
\\
\!\!\!\!E[|\hat b_{hk}|] \!&\geq\!\! \frac{x_{hk}^3}{(\mt{var}[\hat x_{hk}] \!+\! x_{hk}^2)^2 \!+\! (\mt{cov}[\hat x_{hk}, \hat r_{hk}] \!+\! r_{hk}x_{hk})^2}.\!\!\!
\end{align}
\end{subequations}
Because of the consistency of the TLS estimate, $\mt{var}[\hat r_{hk}, \hat x_{hk}] \rightarrow 0$ as $N \rightarrow \infty$. This means that $\hat g_{hk}$ and $\hat b_{hk}$ can be asymptotically unbiased because the RHS of \eqref{eq_gb_expect_def} tends to \eqref{eq_gb_def}. However, with a finite number of samples, $\mt{var}[\hat r_{hk}, \hat x_{hk}]$ can remain quite high, which may heavily bias the admittance estimate.

\section{Network identification}\label{section_network}

Often, DSOs are interested in estimating the model for their full network. Ideally, the model could be estimated from smart meter injection measurements throughout the network. When a grid is radial and the topology is known, the network estimation problem for the full network can be decomposed into individual line parameter estimation problems for each line using Kirchoff's Current Law. However, this is not possible for mesh networks or radial networks with unknown topologies. Thus, it is desirable to have a method for estimating the grid parameters for a full network\footnote{The ``full'' network that is estimated consists of the Kron-reduced network connecting just the nodes at which the injections are measured. Note, the Kron-reduction is different than the statistical reductions that are used elsewhere in this paper. The unmeasured injections are treated as noise \cite{dorfler2018electrical, moffat_2019}.} from just smart meter voltage and injection measurements.

Estimating a network from just smart meter voltage and injection measurements is a much more challenging task than estimating just a single line, however, and thus requires the following additional approximation:
\begin{align*}
    i_h \approx |i_h|e^{-j\phi_h} \hspace{26pt} &\mt{: neglect current phase shifts}.
\end{align*}
Such approximations are quite accurate under Assumption \ref{ass_small_theta}, as shown in \cite{zhang_non_pmu_id}. 
With the aforementioned approximations, the Kirchhoff law at node $h$ is given by $|i_h|e^{-j\phi_h} = \sum_{k=1}^n y_{hk}(|v_h|(1-j\theta_h) - |v_k|(1-j\theta_k))$. In matrix form and for all $h \in \mc V$, this gives
\begin{subequations}\label{eq_net_flow_lin_real}
\begin{align}\label{eq_net_flow_lin_real_cos}
    I^\Re &= (G|V|-B(|V|\Theta)),
    \\
    I^\Im &= -(B|V|+G(|V|\Theta)),
    \label{eq_net_flow_lin_real_sin}
\end{align}
\end{subequations}
where $I^\Re = [|i_h| \cos(\phi_h)]_{h=1}^n$, $I^\Re = [|i_h| \sin(\phi_h)]_{h=1}^n$, $|V| = [|v_h|]_{h=1}^n$, and $ \Theta = [|\theta_h|]_{h=1}^n$. The parameter estimation problem consists in finding $G \in \bb R^{n\times n}$ and $B \in \bb R^{n\times n}$ from noisy measurements of the current magnitudes, voltage magnitudes and power angles.


\subsection{Reduced Regression Model}
Similar to \cite{ognjen_pseudoinverse}, we use pseudo-inverses to avoid the problems caused by the singularity of $Y$ in the absence of shunt elements. Recall that if $Y$ is invertible, its pseudo-inverse is equal to its inverse. We first note that $B G^\dagger G = B$ because the matrices have the same null space\footnote{For this reason, the equality also holds for $G B^\dagger B = G$ and other combinations of these matrices, as well as their sums and products.} defined by the vector of all ones and are both rank $n-1$. Hence, multiplying \eqref{eq_net_flow_lin_real_sin} by $B G^\dagger$ on both sides, one can rewrite \eqref{eq_net_flow_lin_real} as
\begin{align*}
    B(|V|\theta) &= (G|V| - I^\Re),
    \\
    B(|V|\theta) &= -(B G^\dagger B|V|+ B G^\dagger I^\Im).
\end{align*}
Combining the two equations yields the reduced model
\begin{align}\label{eq_net_flow_lin_real_pseudo}
    (G + B G^\dagger B)|V| &=  I^\Re -  B G^\dagger I^\Im.
\end{align}
This model was first presented in \cite{zhang_non_pmu_id}, where the authors regress $|V|$ and $I^\Im$ onto $I^\Re$ to obtain $G = (G + B G^\dagger B)(GG^\dagger + (B G^\dagger)^2)^\dagger$, and similarly for $B$. However, as explained in Section \ref{subsec_endo}, the regression \eqref{eq_net_flow_lin_real_pseudo} suffers from a correlation between $I^\Re$ and $I^\Im$, which makes the estimation of $B G^\dagger$ very challenging.

\subsection{Equivalent Impedance}
In order to avoid the bias caused by the correlation between $I^\Re$ and $I^\Im$, we define the matrix $O$ such that $O_{hh} = \frac{n-1}{n}$ and $O_{hk} = \frac{1}{n}$. This matrix has the same null space as $G$ and $B$ and has rank $n-1$. Moreover, the matrix $(G + B G^\dagger B)$ also has the same null space and rank, which means that $O (G + B G^\dagger B)^\dagger (G + B G^\dagger B) = O$. Moreover, by construction $O (G + B G^\dagger B)^\dagger = (G + B G^\dagger B)^\dagger$. We can therefore multiply both sides of \eqref{eq_net_flow_lin_real_pseudo} by $O (G + B G^\dagger B)^\dagger$ to obtain
\begin{align*}
    O|V| &=  (G + B G^\dagger B)^\dagger I^\Re -  (G + B G^\dagger B)^\dagger B G^\dagger I^\Im.
\end{align*}
Finally, because $(G + B G^\dagger B)^\dagger B G^\dagger = (B + G B^\dagger G)^\dagger$, similar to \eqref{eq_line_flow_z_lin_real_reduced_simple} and \eqref{eq_rx_def}, one obtains the regression
\begin{align}\label{eq_impedance_regression}
    O|V| &=  R I^\Re + X I^\Im,
\end{align}
where
\begin{align}\label{eq_RX_def}
    \begin{array}{ll}
    \,R =  (G + B G^\dagger B)^\dagger,
    \\
    X = -(B + G B^\dagger G)^\dagger,
    \end{array}
    \Leftrightarrow
    \begin{array}{ll}
    G =  (R + X R^\dagger X)^\dagger,
    \\
    B = -(X + R X^\dagger R)^\dagger.
    \end{array}
\end{align}
Estimating $R$ and $X$ using the TLS is consistent, unbiased and efficient for the same reasons as in the single line problem. 
%
Similar to Section \ref{subsec_xr_regression}, the pseudo-inversion of the estimates of $R$ and $X$ lead to an asymptotic bias. It can also be bounded using \cite[Theorem 2.1]{multivariate_jensen}. The exact expressions are very long and are outside of the scope of this paper. Nevertheless, we show the accuracy of the estimates numerically in the next section.

\section{Unbalanced three-phases identification}\label{section_threePhase}
The parameter estimation of unbalanced three-phase networks can be done using the same methods as for the single-phase ones by stacking the voltages and current matrices as $V_{|} = [V_k]_{k=a,b,c} \in \bb R^{3n}$ and similarly for $I^\Re_{|}$ and $I^\Im_{|}$. However, this greatly reduces the SNR as the phases are often heavily correlated with one another. In this case, the voltage difference between the phases is too small to be detected. There are two main ways to simplify the problem:
\begin{enumerate}[label=(\roman*)]
    \item Assume that all self-admittances are equal and that all mutual admittances are equal.
    \item Assume that all self-admittances are equal and that all mutual admittances are zero.
\end{enumerate}

The simplification (i) still tries to estimate the mutual admittances and is adapted for heavily unbalanced loads. It is equivalent to the assumption that the network infrastructure is balanced, while the loads are not. This assumption is similar to neglecting the negative and zero sequences and holds quite well if all phases follow the same path, even if they are not transposed \cite{tsg_bayesian}. The resulting regression can be written as
\begin{align}
    O|V_{abc}| &=  R_L I^\Re_{abc} + X_L I^\Im_{abc} + R_M 
    \\ \nonumber &\quad\quad
    + (I^\Re_{cab} + I^\Re_{bca}) + X_M (I^\Im_{cab} + I^\Im_{bca}) 
\end{align}
where $V_{abc} = [V_a, V_b, V_c] \in \bb R^{n \times 3}$ and similarly for $I^\Re_{abc}$, $I^\Re_{cab}$, $I^\Re_{bca}$, $I^\Im_{abc}$, $I^\Im_{cab}$, and $I^\Im_{bca}$.

The simplification (ii) holds for balanced enough systems, when the influence of mutual components cancels out. The regression \eqref{eq_impedance_regression} can be used with $V_{abc}$, $I^\Re_{abc}$, and $I^\Im_{abc}$ in place of $V$, $I^\Re$, and $I^\Im$, respectively. The number of parameters to estimate is reduced, therefore increasing the data-efficiency. However, it may not be applicable in networks with large single-phase loads.


\section{Experiments} \label{section_experiments}
We validate our results on the 16kV distribution grid of Walenstadt, Switzerland.
Not all the nodes are observed/have smart meter measurements. 
We first use the line from the substation to Schlittriet to test the single-line identification methods. 
We then experiment with network identification using the Br{\"u}sis-to-Freihof six-node sub-network.
A groundtruth model of the network is not available so we assess the accuracy of the estimation methods using a combination of intuition, cross-validation, and synthetic data, as explained below.

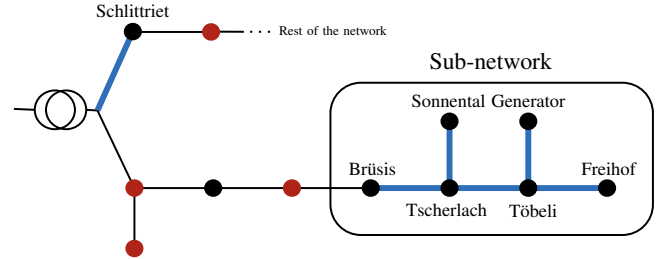
\begin{figure}[h]
    \tikzset{every picture/.style={line width=0.75pt}} 

\begin{tikzpicture}[x=0.75pt,y=0.75pt,yscale=-1,xscale=1]

\draw [color={rgb, 255:red, 47; green, 111; blue, 185 }  ,draw opacity=1 ][line width=2.25]    (187,95.5) -- (226.73,95.5) ;
\draw [color={rgb, 255:red, 47; green, 111; blue, 185 }  ,draw opacity=1 ][line width=2.25]    (226.73,95.5) -- (226.73,61.79) ;
\draw [color={rgb, 255:red, 47; green, 111; blue, 185 }  ,draw opacity=1 ][line width=2.25]    (226.73,95.5) -- (266.46,95.5) ;
\draw [color={rgb, 255:red, 47; green, 111; blue, 185 }  ,draw opacity=1 ][line width=2.25]    (266.46,95.5) -- (306.19,95.5) ;
\draw [color={rgb, 255:red, 47; green, 111; blue, 185 }  ,draw opacity=1 ][line width=2.25]    (266.46,95.5) -- (266.46,61.79) ;
\draw  [fill={rgb, 255:red, 0; green, 0; blue, 0 }  ,fill opacity=1 ] (103.49,95.5) .. controls (103.49,93.34) and (105.3,91.58) .. (107.54,91.58) .. controls (109.78,91.58) and (111.6,93.34) .. (111.6,95.5) .. controls (111.6,97.65) and (109.78,99.41) .. (107.54,99.41) .. controls (105.3,99.41) and (103.49,97.65) .. (103.49,95.5) -- cycle ;
\draw  [fill={rgb, 255:red, 0; green, 0; blue, 0 }  ,fill opacity=1 ] (182.95,95.5) .. controls (182.95,93.34) and (184.76,91.58) .. (187,91.58) .. controls (189.24,91.58) and (191.06,93.34) .. (191.06,95.5) .. controls (191.06,97.65) and (189.24,99.41) .. (187,99.41) .. controls (184.76,99.41) and (182.95,97.65) .. (182.95,95.5) -- cycle ;
\draw  [fill={rgb, 255:red, 0; green, 0; blue, 0 }  ,fill opacity=1 ] (222.68,95.5) .. controls (222.68,93.34) and (224.49,91.58) .. (226.73,91.58) .. controls (228.97,91.58) and (230.79,93.34) .. (230.79,95.5) .. controls (230.79,97.65) and (228.97,99.41) .. (226.73,99.41) .. controls (224.49,99.41) and (222.68,97.65) .. (222.68,95.5) -- cycle ;
\draw  [fill={rgb, 255:red, 0; green, 0; blue, 0 }  ,fill opacity=1 ] (262.41,95.5) .. controls (262.41,93.34) and (264.22,91.58) .. (266.46,91.58) .. controls (268.7,91.58) and (270.52,93.34) .. (270.52,95.5) .. controls (270.52,97.65) and (268.7,99.41) .. (266.46,99.41) .. controls (264.22,99.41) and (262.41,97.65) .. (262.41,95.5) -- cycle ;
\draw  [fill={rgb, 255:red, 0; green, 0; blue, 0 }  ,fill opacity=1 ] (222.68,61.79) .. controls (222.68,59.63) and (224.49,57.88) .. (226.73,57.88) .. controls (228.97,57.88) and (230.79,59.63) .. (230.79,61.79) .. controls (230.79,63.95) and (228.97,65.7) .. (226.73,65.7) .. controls (224.49,65.7) and (222.68,63.95) .. (222.68,61.79) -- cycle ;
\draw  [fill={rgb, 255:red, 0; green, 0; blue, 0 }  ,fill opacity=1 ] (262.41,61.79) .. controls (262.41,59.63) and (264.22,57.88) .. (266.46,57.88) .. controls (268.7,57.88) and (270.52,59.63) .. (270.52,61.79) .. controls (270.52,63.95) and (268.7,65.7) .. (266.46,65.7) .. controls (264.22,65.7) and (262.41,63.95) .. (262.41,61.79) -- cycle ;
\draw  [fill={rgb, 255:red, 0; green, 0; blue, 0 }  ,fill opacity=1 ] (302.14,95.5) .. controls (302.14,93.34) and (303.95,91.58) .. (306.19,91.58) .. controls (308.43,91.58) and (310.25,93.34) .. (310.25,95.5) .. controls (310.25,97.65) and (308.43,99.41) .. (306.19,99.41) .. controls (303.95,99.41) and (302.14,97.65) .. (302.14,95.5) -- cycle ;
\draw   (23.22,56) .. controls (23.22,50.6) and (27.75,46.22) .. (33.35,46.22) .. controls (38.95,46.22) and (43.49,50.6) .. (43.49,56) .. controls (43.49,61.4) and (38.95,65.78) .. (33.35,65.78) .. controls (27.75,65.78) and (23.22,61.4) .. (23.22,56) -- cycle ;
\draw   (17.54,56) .. controls (17.54,50.6) and (22.08,46.22) .. (27.68,46.22) .. controls (33.27,46.22) and (37.81,50.6) .. (37.81,56) .. controls (37.81,61.4) and (33.27,65.78) .. (27.68,65.78) .. controls (22.08,65.78) and (17.54,61.4) .. (17.54,56) -- cycle ;
\draw    (17.54,56) -- (7,55.54) ;
\draw    (49.16,56.32) -- (43.49,56) ;
\draw [color={rgb, 255:red, 47; green, 111; blue, 185 }  ,draw opacity=1 ][fill={rgb, 255:red, 0; green, 0; blue, 0 }  ,fill opacity=1 ][line width=2.25]    (67,16.51) -- (49.16,56.32) ;
\draw    (67.81,95.5) -- (49.16,56.32) ;
\draw    (67.81,126) -- (67.81,95.5) ;
\draw    (67.81,95.5) -- (107.54,95.5) ;
\draw  [color={rgb, 255:red, 175; green, 36; blue, 23 }  ,draw opacity=1 ][fill={rgb, 255:red, 175; green, 36; blue, 23 }  ,fill opacity=1 ] (63.76,95.5) .. controls (63.76,93.34) and (65.57,91.58) .. (67.81,91.58) .. controls (70.05,91.58) and (71.87,93.34) .. (71.87,95.5) .. controls (71.87,97.65) and (70.05,99.41) .. (67.81,99.41) .. controls (65.57,99.41) and (63.76,97.65) .. (63.76,95.5) -- cycle ;
\draw  [color={rgb, 255:red, 175; green, 36; blue, 23 }  ,draw opacity=1 ][fill={rgb, 255:red, 175; green, 36; blue, 23 }  ,fill opacity=1 ] (63.76,126) .. controls (63.76,123.84) and (65.57,122.09) .. (67.81,122.09) .. controls (70.05,122.09) and (71.87,123.84) .. (71.87,126) .. controls (71.87,128.16) and (70.05,129.91) .. (67.81,129.91) .. controls (65.57,129.91) and (63.76,128.16) .. (63.76,126) -- cycle ;
\draw    (107.54,95.5) -- (147.27,95.5) ;
\draw    (147.27,95.5) -- (187,95.5) ;
\draw   (166.67,57.7) .. controls (166.67,49.21) and (173.55,42.33) .. (182.04,42.33) -- (313.94,42.33) .. controls (322.43,42.33) and (329.31,49.21) .. (329.31,57.7) -- (329.31,103.81) .. controls (329.31,112.3) and (322.43,119.18) .. (313.94,119.18) -- (182.04,119.18) .. controls (173.55,119.18) and (166.67,112.3) .. (166.67,103.81) -- cycle ;
\draw  [color={rgb, 255:red, 175; green, 36; blue, 23 }  ,draw opacity=1 ][fill={rgb, 255:red, 175; green, 36; blue, 23 }  ,fill opacity=1 ] (143.22,95.5) .. controls (143.22,93.34) and (145.03,91.58) .. (147.27,91.58) .. controls (149.51,91.58) and (151.33,93.34) .. (151.33,95.5) .. controls (151.33,97.65) and (149.51,99.41) .. (147.27,99.41) .. controls (145.03,99.41) and (143.22,97.65) .. (143.22,95.5) -- cycle ;
\draw  [fill={rgb, 255:red, 0; green, 0; blue, 0 }  ,fill opacity=1 ] (62.95,16.51) .. controls (62.95,14.35) and (64.76,12.6) .. (67,12.6) .. controls (69.24,12.6) and (71.05,14.35) .. (71.05,16.51) .. controls (71.05,18.67) and (69.24,20.42) .. (67,20.42) .. controls (64.76,20.42) and (62.95,18.67) .. (62.95,16.51) -- cycle ;
\draw    (123,16.61) -- (105.8,16.51) ;
\draw  [dash pattern={on 0.84pt off 2.51pt}]  (136.2,16.61) -- (123,16.61) ;
\draw    (66.74,16.7) -- (106.47,16.7) ;
\draw  [color={rgb, 255:red, 175; green, 36; blue, 23 }  ,draw opacity=1 ][fill={rgb, 255:red, 175; green, 36; blue, 23 }  ,fill opacity=1 ] (102.42,16.7) .. controls (102.42,14.54) and (104.23,12.78) .. (106.47,12.78) .. controls (108.71,12.78) and (110.53,14.54) .. (110.53,16.7) .. controls (110.53,18.85) and (108.71,20.61) .. (106.47,20.61) .. controls (104.23,20.61) and (102.42,18.85) .. (102.42,16.7) -- cycle ;

\draw (47.04,1.43) node [anchor=north west][inner sep=0.75pt]  [font=\scriptsize] [align=left] {Schlittriet};
\draw (174.46,80.87) node [anchor=north west][inner sep=0.75pt]  [font=\scriptsize] [align=left] {Brüsis};
\draw (203.3,101.5) node [anchor=north west][inner sep=0.75pt]  [font=\scriptsize] [align=left] {Tscherlach};
\draw (206.1,46.84) node [anchor=north west][inner sep=0.75pt]  [font=\scriptsize] [align=left] {Sonnental};
\draw (255.29,102.17) node [anchor=north west][inner sep=0.75pt]  [font=\scriptsize] [align=left] {Töbeli};
\draw (246.35,46.84) node [anchor=north west][inner sep=0.75pt]  [font=\scriptsize] [align=left] {Generator};
\draw (291.8,81.2) node [anchor=north west][inner sep=0.75pt]  [font=\scriptsize] [align=left] {Freihof};
\draw (215.1,25.84) node [anchor=north west][inner sep=0.75pt]  [font=\small] [align=left] {Sub-network};
\draw (139.24,12.4) node [anchor=north west][inner sep=0.75pt]  [font=\tiny] [align=left] {Rest of the network};

\end{tikzpicture}
    \caption{Graph of the region of interest in the Walenstadt network. 
    The lines that will be identified are in blue. The unobserved nodes are in red.}
    \label{fig_network}
\end{figure}

\subsection{Single Line Estimation}\label{subsec_single_results}
\subsubsection{Real-World Data Single Line Estimation}
The line of interest connecting "Schlittriet" to the substation (see Fig. \ref{fig_network}) has a resistance of approximately $0.1\Omega$.\footnote{The linear resistivity is $114$m$\Omega/$km and the line is $880$m long. The reactance is much harder to compute as it depends on the environment of the line \cite{carson1927electromagnetic}.} Its reactance is however unknown. 
We first use the raw data into all four methods described in Section \ref{section_biases} to obtain the results in Fig. \ref{fig_real_results}. In the rest of the section, we use the sub-figure indices (a)-(d) of Fig. \ref{fig_real_results} to denote the corresponding methods.

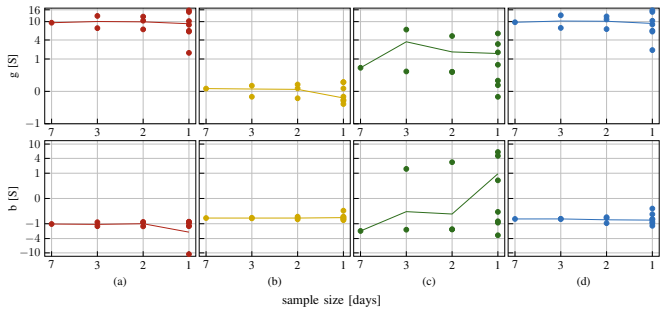
\begin{figure}[h]
\begin{tikzpicture}[scale=0.45]
\pgfplotsset{/pgfplots/group/.cd,
    horizontal sep=0.1cm,
    vertical sep=0.5cm
}
\pgfplotsset{
    y coord trafo/.code={\pgfmathparse{ifthenelse(#1<0,
    -(abs(#1))^(1/3) ,
    abs(#1)^(1/3)
    )}},
    y coord inv trafo/.code={\pgfmathparse{(#1)^3}}
 }

\definecolor{color0}{HTML}{AF2417}
\definecolor{color1}{HTML}{D0A900}
\definecolor{color2}{HTML}{2E6D1C}
\definecolor{color3}{HTML}{2F6FB9}

\begin{groupplot}[group style={group size=4 by 2}, width=\textwidth/3, height=5cm, max space between ticks=20pt,
ylabel absolute, every axis y label/.append style={yshift=-0.3cm}]


\nextgroupplot[
tick pos=left,
xmajorgrids,
xmin=-0.15, xmax=3.15,
xtick style={color=black},
xtick={0,1,2,3},
xticklabels={7,3,2,1},
ylabel={g [S]},
ymajorgrids, ytick={-16,-10,-4,-1,0,1,4,10,16},
ymin=-1, ymax=17,
ytick style={color=black}
]
\addplot [draw=color0, fill=color0, mark=*, only marks]
table{%
x  y
0 9.54357612540022
1 12.6574135048312
1 7.44605520428778
2 10.4036328519783
2 12.319900647989
2 7.02619056490721
3 1.68379967259779
3 10.2204110917906
3 15.6268383097752
3 14.7464822350758
3 6.28998766513944
3 8.70442918876291
3 6.54826705910609
};
\addplot [semithick, color0]
table {%
0 9.54357612540022
1 10.0517343545595
2 9.91657468829152
3 9.11717360317826
};

\nextgroupplot[
tick pos=left,
xmajorgrids,
xmin=-0.15, xmax=3.15,
xtick style={color=black},
xtick={0,1,2,3},
xticklabels={7,3,2,1},
ymajorgrids, ytick={-16,-10,-4,-1,0,1,4,10,16}, yticklabel=\empty,
ymin=-1, ymax=17,
ytick style={color=black}
]
\addplot [draw=color1, fill=color1, mark=*, only marks]
table{%
x  y
0 0.000638837310729344
1 0.00539652124142492
1 -0.00462251550158148
2 -0.00967089246595499
2 0.0096045514144248
2 0.000790043683874281
3 -0.0593720665662654
3 -0.00420853036660674
3 0.0232676023864187
3 0.000729838236809925
3 -0.0233744720659089
3 -0.0185138699873914
3 0.0225418019589374
};
\addplot [semithick, color1]
table {%
0 0.000638837310729344
1 0.000387002869921719
2 0.000241234210781366
3 -0.0084185280577152
};

\nextgroupplot[
tick pos=left,
xmajorgrids,
xmin=-0.15, xmax=3.15,
xtick style={color=black},
xtick={0,1,2,3},
xticklabels={7,3,2,1},
ymajorgrids, ytick={-16,-10,-4,-1,0,1,4,10,16}, yticklabel=\empty,
ymin=-1, ymax=17,
ytick style={color=black}
]
\addplot [draw=color2, fill=color2, mark=*, only marks]
table{%
x  y
0 0.38756661620693
1 0.232791089188446
1 6.95219193738733
2 0.212503283752191
2 0.222130807648117
2 4.99696174813751
3 1.75356668330252
3 5.72381643912189
3 3.12749502510901
3 0.555893076582948
3 -0.00468572169651957
3 0.00702828977281402
3 0.0366481039208962
};
\addplot [semithick, color2]
table {%
0 0.38756661620693
1 3.59249151328789
2 1.81053194651261
3 1.59996598515908
};

\nextgroupplot[
tick pos=left,
xmajorgrids,
xmin=-0.15, xmax=3.15,
xtick style={color=black},
xtick={0,1,2,3},
xticklabels={7,3,2,1},
ymajorgrids, ytick={-16,-10,-4,-1,0,1,4,10,16}, yticklabel=\empty,
ymin=-1, ymax=17,
ytick style={color=black}
]
\addplot [draw=color3, fill=color3, mark=*, only marks]
table{%
x  y
0 9.7384249826778
1 12.9878041024377
1 7.52924610294814
2 10.9544717008402
2 12.4075119614076
2 7.11943831543386
3 2.07217542918334
3 10.3711651765617
3 15.7126618979316
3 14.7999460745488
3 6.30003249835157
3 8.82057674798427
3 6.58423256024556
};
\addplot [semithick, color3]
table {%
0 9.7384249826778
1 10.2585251026929
2 10.1604739925606
3 9.23725576925813
};

\nextgroupplot[
tick pos=left,
xlabel={(a)},
xmajorgrids,
xmin=-0.15, xmax=3.15,
xtick style={color=black},
xtick={0,1,2,3},
xticklabels={7,3,2,1},
ylabel={b [S]},
ymajorgrids, ytick={-16,-10,-4,-1,0,1,4,10,16},
ymin=-12, ymax=12,
ytick style={color=black}
]
\addplot [draw=color0, fill=color0, mark=*, only marks]
table{%
x  y
0 -1.03563494646558
1 -1.32499477207647
1 -0.831381471564943
2 -1.3465640663814
2 -0.809544984425639
2 -0.838839960695075
3 -1.25603393822687
3 -1.3220327884153
3 -10.7318289176668
3 -0.827100007307507
3 -0.775632444123561
3 -0.834703072195239
3 -0.841098407386889
};
\addplot [semithick, color0]
table {%
0 -1.03563494646558
1 -1.07818812182071
2 -0.998316337167373
3 -2.36977565361746
};

\nextgroupplot[
tick pos=left,
xlabel={(b)},
xmajorgrids,
xmin=-0.15, xmax=3.15,
xtick style={color=black},
xtick={0,1,2,3},
xticklabels={7,3,2,1},
ymajorgrids, ytick={-16,-10,-4,-1,0,1,4,10,16}, yticklabel=\empty,
ymin=-12, ymax=12,
ytick style={color=black}
]
\addplot [draw=color1, fill=color1, mark=*, only marks]
table{%
x  y
0 -0.465768266733595
1 -0.487060966223934
1 -0.446438217985333
2 -0.373471413369638
2 -0.563797386342945
2 -0.459041416675826
3 -0.112260286366174
3 -0.439730917877796
3 -0.633620984858824
3 -0.553125551499153
3 -0.367828848989577
3 -0.406929972567258
3 -0.536417570625598
};
\addplot [semithick, color1]
table {%
0 -0.465768266733595
1 -0.466749592104633
2 -0.465436738796136
3 -0.435702018969197
};

\nextgroupplot[
tick pos=left,
xlabel={(c)},
xmajorgrids,
xmin=-0.15, xmax=3.15,
xtick style={color=black},
xtick={0,1,2,3},
xticklabels={7,3,2,1},
ymajorgrids, ytick={-16,-10,-4,-1,0,1,4,10,16}, yticklabel=\empty,
ymin=-12, ymax=12,
ytick style={color=black}
]
\addplot [draw=color2, fill=color2, mark=*, only marks]
table{%
x  y
0 -2.12903338071442
1 -1.88863807620154
1 1.60548699355862
2 -1.83029915580433
2 -1.83255349469051
2 2.96400594711467
3 -0.149831868090641
3 4.80741867555744
3 6.21527966899918
3 -3.09451601535966
3 -0.737198531381987
3 -0.880498873655695
3 0.369889975318642
};
\addplot [semithick, color2]
table {%
0 -2.12903338071442
1 -0.141575541321461
2 -0.232948901126723
3 0.932934718769612
};

\nextgroupplot[
tick pos=left,
xlabel={(d)},
xmajorgrids,
xmin=-0.15, xmax=3.15,
xtick style={color=black},
xtick={0,1,2,3},
xticklabels={7,3,2,1},
ymajorgrids, ytick={-16,-10,-4,-1,0,1,4,10,16}, yticklabel=\empty,
ymin=-12, ymax=12,
ytick style={color=black}
]
\addplot [draw=color3, fill=color3, mark=*, only marks]
table{%
x  y
0 -0.532905301534644
1 -0.521935261817151
1 -0.538905067842831
2 -0.935756828137912
2 -0.408471720900493
2 -0.460575208818599
3 -1.26393297528679
3 -0.547470611037871
3 -0.063999550084363
3 -0.564058164157673
3 -0.786790215145341
3 -0.932294629807294
3 -0.244600313400803
};
\addplot [semithick, color3]
table {%
0 -0.532905301534644
1 -0.530420164829991
2 -0.601601252619002
3 -0.629020922702876
};
\end{groupplot}

\draw ({$(current bounding box.south west)!0.5!(current bounding box.south east)$}|-{$(current bounding box.south west)!-0.07!(current bounding box.north west)$}) node[
  scale=0.495,
  anchor=south,
  text=black,
  rotate=0.0
]{ sample size [days]};
\end{tikzpicture}
    \caption{Estimates of the conductance and susceptance of a single line using (a) the admittance regression \eqref{eq_simple_ls_regression_noisy}, (b) the joint estimation using \eqref{eq_simple_ls_regression_noisy_phase}, (c) the reduced model
\eqref{eq_line_flow_lin_real_reduced}, (d) the impedance regression
\eqref{eq_line_flow_z_lin_real_reduced_simple} with real data. For each method, the points mark the estimates using each data subset and the line their average. 
The true conductance $b$ and susceptance $g$ of the line is not known, but $R \approx 0.1 \Omega$.} 
    \label{fig_real_results}
\end{figure}

We assess the consistency of the estimates with $R = 0.1\Omega$, i.e., our only ground truth, by inverting the 7-days estimate using \eqref{eq_rx_def}. We observe that (a), (c), and (d) are consistent with $R = 0.1\Omega$, but (b) yields $y = 0.003 -0.46j$ S, which corresponds to an incorrect $R = 0.01\Omega$. Moreover, (b) and (c) estimates a much larger X/R ratio than all three other methods. 

Because the true conductance and susceptance of the line is not known, we use cross validation in order to determine the accuracy of the methods. We split the data into 6 days for estimating the parameters, and 1 day to validate the line flow power predictions that use the parameters.
Table \ref{table_crossvalid} shows that (a) provides a good estimate for only $P$, (b) provides good estimates for only $Q$, and (c) does not provide a good estimate. Thus, only (d) could approximate both $P$ and $Q$ correctly.

\begin{table}[h]
    \centering
    \begin{tabular}{|cc|c|c|}
        \hline
        & method & P [kW] & Q [kVA]
        \\ \hline
        & measured & 107.49 & 5.19 
        \\
        (a) & admittance \eqref{eq_simple_ls_regression_noisy} & 138.69 & 13.34
        \\
        (b) & joint \eqref{eq_simple_ls_regression_noisy_phase} & 0.044 & 6.98
        \\
        (c) & reduced \eqref{eq_line_flow_lin_real_reduced} & 15.62 & 47.08
        \\
        (d) & impedance \eqref{eq_line_flow_z_lin_real_reduced_simple} & 140.34 & 6.63
        \\ \hline
    \end{tabular}
    \vspace{2pt}
    \caption{Measured and predicted average power transmission in the estimated line for the last day of the data.}
    \label{table_crossvalid}
\end{table}

\subsubsection{Synthetic Data Single Line Estimation}
In order to understand better the very wide range of estimates, we generate synthetic data by computing the voltages that would perfectly match the measured power for all 7 days and $y = 10 - 0.5j$ S. 
We then add an increasing amount of Gaussian noise to observe how the estimates behave and compute the expected biases from Section \ref{section_biases}.

\begin{figure}[h]
    \input{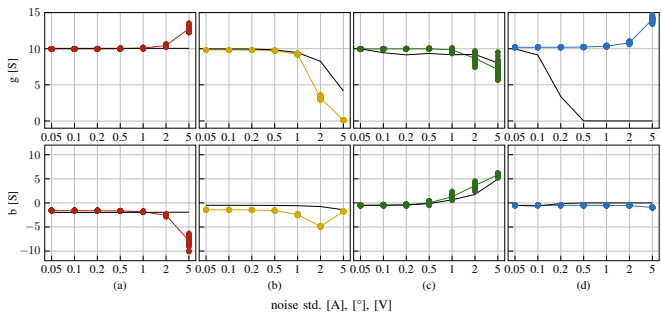}
    \caption{Estimates of the conductance and susceptance of a single line using the four methods (a)-(d) with synthetic data. For each method, the points mark the estimates affected by one of the 50 realizations of the noise generated for each noise level. The line in color shows their average and the black line shows the predicted bias.  In case (d) the black line is a lower bound on the bias.}
    \label{fig_synthetic_results}
\end{figure}

Fig. \ref{fig_synthetic_results} shows that all methods perform much better with synthetic data. However, the $g$ estimates in (b) and (c) get smaller with more noise and the $b$ estimate in (a) is too large by a factor of $2$. Both of these errors can also be observed in Fig. \ref{fig_real_results}. Moreover, the $b$ estimate in (c) changes sign and becomes larger. We also observe that the bias predictions become worse with large amounts of noise, probably because Assumption \ref{ass_small_deltaphi} does not hold anymore. Moreover, the approximation of $E[b]$ for (b) given in Section \ref{subsec_simultaneity} is not very close, and the lower bound \eqref{eq_gb_expect_def} for (d) is very loose in this particular case.

\subsection{Network Estimation}
\subsubsection{Real-World Data Network Estimation}
We use V, P and Q measurements from the sub-network indicated in Fig. \ref{fig_network} in order to estimate its admittance matrix. The singular values of $[V, I^\Re, I^\Im]$ reveal a SNR of approximately 1, which makes the estimation challenging. In order to avoid numerical stability issues\footnote{In high SNR scenarios, the TLS may fail to estimate the EIV correctly and remove parts of the signal instead. While the OLS is biased, it is more robust than the TLS.}, we test the TLS-based methods with OLS as well. The estimation with all 7 days of data yields the following results for each method.
\begin{enumerate}[label=(\alph*)]
    \item from \cite{lisa2023} diverges with both OLS and TLS, yielding a quasi-infinite admittance estimate.
    \item from \cite{zhang_non_pmu_newton} converges to zero. 
    \item from \cite{zhang_non_pmu_id} diverges with the TLS but results in the matrix shown in Fig. \ref{fig_heatmaps} (c).
    \item using \eqref{eq_impedance_regression} yields the matrix shown in Fig. \ref{fig_heatmaps} (d) with the TLS and a similar one with the OLS.
\end{enumerate}
\begin{figure}[h]
     \centering
     \setcounter{subfigure}{2}
     \subfloat[]{
        \includegraphics[width=0.227\textwidth]{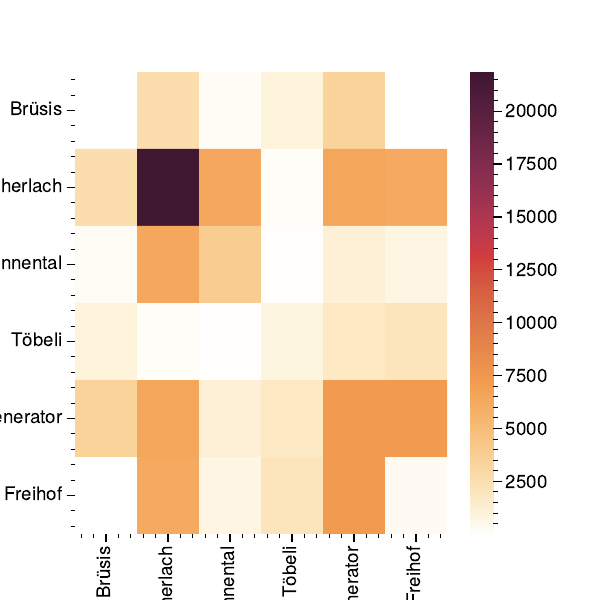}
        }
     \subfloat[]{
        \includegraphics[width=0.227\textwidth]{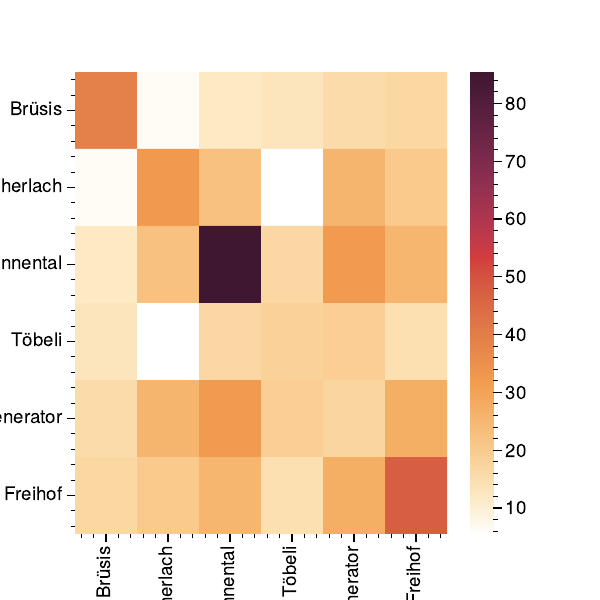}
        }
    \caption{Heatmaps of the magitudes of the admittance matrix estimates given by the methods (c) and (d).}
    \label{fig_heatmaps}
\end{figure}
The resistance of the lines should be similar to the one in Section \ref{subsec_single_results} given their length and cross-section. We note that only (d) is consistent with this observation. While the topology is not recognizable from Fig. \ref{fig_heatmaps} (d), the estimate may be useful for control and state estimation purposes. If it is not, the the data must be improved by installing more accurate smart meter sensors, PMUs, or lineflow measurements.

\section{Conclusion} \label{section_conclu}


Estimating electrical grid parameters using unsychronized measuring devices is challenging. This paper quantifies this challenge by describing the non-zero expected admittance estimation error for four admittance estimation methods in the literature, as well as the variance of impedance estimation. Impedance estimation can be conducted without bias using TLS. The results in this paper warrant further impedance and admittance estimation experiments with real smart meter data, and experiments that combine smart meter data with data from other types of sensors such as uPMUs and lineflow measurements.

\section*{Acknowledgement}
We thank the Swiss National Science Foundation for supporting this research under the NCCR Automation project (grant number 51NF40 180545). We would also like to thank Benjamin Sawicki and WEW (the DSO of Walenstadt) for providing the data required by our experiments. Additionally, we thank Lisa Laurent and Paul Den\'e for the eye-opening results that they found during their Master projects.

\bibliographystyle{IEEEtran}
\bibliography{references}

\appendices

\section{Dilution Bias and Total Least Squares}\label{app_TLS}

We explore different methods to fit regression models such as \eqref{eq_line_flow_z_lin_real_reduced_simple} or \eqref{eq_simple_ls_regression_noisy_a}. In order to present the theory in various cases, we consider a general regression model $\tilde z = Ax + \epsilon$. For example, in \eqref{eq_line_flow_z_lin_real_reduced_simple} $z$ represents the voltage drop and $x$ represents the current. The matrix $A$ contains the parameters and the noise term $\epsilon$ is such that $E[\epsilon] = 0$. 

\subsection{Dilution Bias}
The OLS are a very common method for parameter estimation, which estimates $A$ as 
\begin{align}\label{eq_app_ols_def}
    \hat A = \mt{var}[x]^{-1} \mt{cov}[x, \tilde z].
\end{align}
This estimate is exact if (i) the variance and covariance estimates are exact and (ii) $\mt{cov}[x, \epsilon] = 0$, which means $\mt{cov}[x, \tilde z] = A \mt{var}[x]$. The condition (i) is usually almost satisfied with very large amounts of data, as the uncertainty on variance scales with the inverse of the number of samples. 

The condition (ii) is not satisfied in \eqref{eq_line_flow_z_lin_real_reduced_simple} and \eqref{eq_simple_ls_regression_noisy_a}, as the noises on $|\tilde v_h| - |\tilde v_k|$,  $|\tilde i_{hk}|\cos(\tilde \phi_{hk})$, and $|\tilde i_{hk}|\sin(\tilde \phi_{hk})$ are collected in $\epsilon_{hk}$, thus creating correlations. Independent uncertainty in the measurement $\tilde x$ of $x$, then $\mt{var}[\tilde x] = \mt{var}[x] + \mt{var}[\tilde x - x]$, so $\mt{var}[\tilde x]^{-1} \prec \mt{var}[x]^{-1}$, which leads to a dilution bias in the OLS \cite{markovsky_tls_overview}.

\subsection{Total Least Squares}\label{app_tls_bias} 
The TLS estimate the parameters $A$ and the exact variance $\mt{var}[x]$ simultaneously, assuming that $\epsilon = \epsilon_z + \epsilon_x \sim \mc N(0,\Sigma_z) + \mc N(0,A \Sigma_z A^\top)$, i.e., that the right and left hand side variables are i.i.d., Gaussian and centered on zero. The variance $\mt{var}[x]$ is computed as $\mt{var}[\tilde x] - \mt{var}[\epsilon_x]$, where $\mt{var}[\epsilon_x]$ is obtained using a low-rank approximation of $[x + \epsilon_x, z + \epsilon_z]$ \cite{markovsky_tls_overview}\footnote{The TLS can also be computed online using its recursive formulation \cite{rhode_rgtls}}.  In this case, \eqref{eq_app_ols_def} becomes
\begin{align}\label{eq_app_tls_def}
    \hat A = (\mt{var}[\tilde x] - \mt{var}[\epsilon_x])^{-1} \mt{cov}[\tilde x, \tilde z].
\end{align}
This estimate is unbiased if and only if $\mt{cov}[\tilde x, \tilde z] = \mt{cov}[x, z] + \mt{cov}[\epsilon_x, \tilde z] + \mt{cov}[\tilde x, \epsilon_z] - \mt{cov}[\epsilon_x, \epsilon_z] = \mt{cov}[x, z]$. Otherwise the bias is given by $(\mt{var}[\tilde x] - \mt{var}[\epsilon_x])^{-1} (\mt{cov}[\tilde x, \tilde z] - \mt{cov}[x, z])$, which is equal to \eqref{eq_tls_bias} when $D = 1$.

\subsection{Multiple Datasets}
If $D$ datasets $\tilde x_d, \tilde z_d$ of equal size such that $z_d = A x_d$ for $d = 1,\dots, D$ are combined, \eqref{eq_app_ols_def} becomes $\hat A = \big(\sum_{d=1}^D \mt{var}[x_d]\big)^{-1}\sum_{d=1}^D \mt{cov}[x_d, \tilde z_d]$. Moreover, similar to \eqref{eq_app_tls_def}, the TLS estimate uses $\mt{var}[x_d] = \mt{var}[\tilde x_d] - \mt{var}[\tilde x_d - x_d]$ when the noise is independent. Hence the bias is $\big(\sum_{d=1}^D \mt{var}[x_d] - \mt{var}[\tilde x_d - x_d]\big)^{-1}\sum_{d=1}^D (\mt{cov}[\tilde x_d, \tilde z_d] - \mt{cov}[x_d, z_d])$, which is equal to \eqref{eq_tls_bias}.




%



\end{document}